\documentclass[a4paper,english]{lipics}

\usepackage{amsmath}
\usepackage{amssymb}

\usepackage{ltlfonts}
\usepackage{xspace}
\usepackage{paralist}

\usepackage{cite} 

\newcommand{\comment}[1]{}

\newcommand{\mc}[1]{\mathcal{#1}}

\newcommand{\Out}{\ensuremath{\mc{V}_\mc{O}} }
\newcommand{\out}{\ensuremath{\mathsf{out}} }

\newcommand{\AltPaths}{\mathsf{AltPaths}}

\newcommand{\false}{\mathsf{false}}

\newcommand{\true}{\mathsf{true}}

\newcommand{\Release}{\mathbin\mathcal{R}}

\newcommand{\tupleof}[1]{\ensuremath{\langle #1 \rangle}}
\newcommand{\Until}{\mathbin{\mathcal{U}}}

\newcommand{\V}{\ensuremath{\mathcal{V}} }
\newcommand{\I}{\ensuremath{\mathcal{I}} }

\newcommand{\U}{\ensuremath{\:\mathcal{U}\hspace{1pt}} }
\newcommand{\W}{\ensuremath{\mathcal{W}} }
\newcommand{\R}{\ensuremath{\hspace{1pt}\mathcal{R}\hspace{1pt}} }

\newcommand{\AP}{\ensuremath{\mathsf{AP}} }

\newcommand{\Hide}{\ensuremath{\mathcal{H}\hspace{1pt}} }

\newcommand{\DO}{\mathsf{do}}
\renewcommand{\Out}{\mathit{Out}}

\renewcommand{\input}{\mathsf{in}}

\newcommand{\Paths}{\ensuremath{\mathsf{Paths}} }

\newcommand{\trace}{\ensuremath{\mathit{trace}} }

\begin{document}

\title{A Temporal Logic for Hyperproperties}

\author[1]{Bernd Finkbeiner}
\author[1]{Markus N. Rabe}
\author[2,3]{C\'esar S\'anchez}
\affil[1]{Saarland University, 
  Saarbr\"ucken, Germany\\
  \texttt{\{finkbeiner,rabe\}@cs.uni-saarland.de}}
\affil[2]{IMDEA Software Institute, Madrid, Spain\\
  \texttt{cesar.sanchez@imdea.org}}
\affil[3]{Institute for Information Security, CSIC, Spain}



\maketitle

\begin{abstract}
  Hyperproperties, as introduced by Clarkson and Schneider,
  characterize the correctness of a computer program as a condition on
  its \emph{set} of computation paths.  Standard temporal logics can
  only refer to a single path at a time, and therefore cannot express many
  hyperproperties of interest, including noninterference and other
  important properties in security and coding theory. In this paper,
  we investigate an extension of temporal logic with explicit path
  variables. We show that the quantification over paths naturally
  subsumes other extensions of temporal logic with operators for
  information flow and knowledge. The model checking problem for
  temporal logic with path quantification is decidable. For
  alternation depth 1, the complexity is PSPACE in the length of the
  formula and NLOGSPACE in the size of the system, as for linear-time
  temporal logic.
  \hspace{18.5em}\fbox{\textbf{Submitted to CSL'13 on April 15, 2013}}
\end{abstract}

\section{Introduction}

The automatic verification of computer systems against specifications
in temporal logic is one of the great success stories of logic in
computer science. A controversial question of this area, with a
continuous discussion in the literature since the early 1980s, is
which temporal logic is most suitable to describe and verify the
relevant properties of interest.
Traditionally, this discussion has been framed as a choice between
linear and branching time~\cite{Vardi/2011/Branching}. The linear time
paradigm of logics like linear-time temporal logic (LTL) describes the
correct behavior as a set of
sequences~\cite{Pnueli/1977/TheTemporalLogicOfPrograms}. This view
allows us to express important temporal properties, like invariants
and eventualities, but ignores possible dependencies between different
executions of the system. The branching time paradigm of logics like
computation tree logic (CTL) describes the correct behavior in terms
of multiple possible futures and therefore allows to express of the
existence of certain
computations~\cite{Clarke+Emerson/1981/CTL}. Emerson and Halpern's
CTL*
logic~\cite{Emerson+Halpern/1986/BranchingVersusLinearTimeTemporalLogic}
unifies the linear and branching time view by combining path formulas
with the path \emph{quantifiers} E and A, which allows to continue
on an existentially or universally chosen path.

An important limitation, common to all these temporal logics, is that
they only refer to a \emph{single} path at a time and therefore cannot
express requirements that relate \emph{multiple} paths. In particular,
this limitation excludes \emph{information flow} properties as they
are studied in security and in coding theory. For example,
\emph{noninterference}~\cite{Goguen+Meseguer/1982/SecurityPoliciesAndSecurityModels},
the property that certain secrets remain hidden from an observer,
relates \emph{all} paths that have (possibly) different secrets but
otherwise identical input. Noninterference requires that all such
paths are observationally equivalent to each other. In coding theory,
a typical requirement is that the encoder of a communication
protocol \emph{preserves the entropy} of the input, which means that
there does not exist a \emph{pair of paths} with different input but
identical encoding. Clarkson and Schneider coined the term
\emph{hyperproperties}~\cite{Clarkson+Schneider/10/Hyperproperties},
defined as sets of legal sets of behaviors, for this more general
class of properties. That is, a correct system must produce exactly
one of the legal combinations of behaviors described by the property.

In this paper we study how to extend temporal logics to
hyperproperties. We generalize the path quantifiers
E and A to $\exists \pi$ and $\forall \pi$, respectively, where $\pi$
is a \emph{path variable} that may be referred to in the subsequent
path formula. As a result, we obtain a generalization of CTL*, which
we call HyperCTL. HyperCTL provides a uniform logical framework for
linear-time and branching-time hyperproperties. In this paper, we show
that the quantification over paths naturally subsumes related
extensions of temporal logic that have previously been studied in the
literature, such as the \emph{hide} operator of
SecLTL~\cite{Dimitrova+Finkbeiner+Kovacs+Rabe+Seidl/12/SecLTL}, and
the \emph{knowledge} operator of temporal epistemic
logic~\cite{Meyden/1993/AxiomsForKnowledgeAndTimeInDistributedSystemsWithPerfectRecall,Fagin+Halpern+Moses+Vardi/1995/ReasoningAboutKnowledgeBook}. 
Additionally, HyperCTL can express a rich set of properties that have
previously been studied intensively, particularly in the security
community, but as isolated properties. HyperCTL provides a common logical
framework that allows to customize and compose these properties.


The paper is structured as follows. We introduce HyperCTL in
Section~\ref{sect:logic}. In Section~\ref{sect:relations}, we relate
HyperCTL to other logics and show that the quantification over paths
subsumes operators for information flow and knowledge.  The
\emph{model checking} and \emph{satisfiability} problems of HyperCTL
are discussed in Section~\ref{sect:MCandSat}.
In Section~\ref{sect:ExamplesSecurity}, we show applications of HyperCTL in \emph{security}, including noninterference, declassification, and quantitative information flow.
We conclude with a discussion of open problems in Section~\ref{sect:conclusions}.




\section{HyperCTL}
\label{sect:logic}

In this section, we introduce HyperCTL, our extension of the temporal
logic CTL* for hyperproperties. We begin with a quick review of
standard notation for Kripke structures.

\subsection{Preliminaries: Kripke structures}
\label{sect:preliminaries}

A \emph{Kripke structure} $K=(S,s_0,\delta,\AP,L)$ consists of a set
of \emph{states} $S$, an \emph{initial state} $s_0\in S$, a transition
function $\delta:S\to 2^{S}$, a set of \emph{atomic propositions}
$\AP$, and a labeling function $L:S\to 2^{\AP}$.  We require that all
states $s\in S$ have at least one successor,
$\delta(s)\not=\emptyset$.  The \emph{cardinality} or \emph{size} of a
Kripke structure is defined as the cardinality of its states and
transitions $|K|= |S|\cdot|\AP|+|\delta|$, where $\delta$ is
interpreted as a set of tuples.
A \emph{path} $\pi\in (S\times2^\AP)^{\omega}$ is an infinite sequence
of labeled states (i.e., states equipped with a set of atomic
propositions).
For a path $\pi$ and $i \in \mathbb N$, $\pi(i)$ is the $i$-th state
of the sequence and $\pi[i]$ is the $i$-th set of atomic propositions.
Furthermore, we use $\pi[0,i]$ to denote the prefix of $\pi$ up to
(including) position $i$, and $\pi[i,\infty]$ to denote the suffix starting
at position $i$.  We lift these notions to tuples of paths by applying
the operations for each position:
e.g. $\langle\pi_1,\dots,\pi_n\rangle[i,\infty]=\langle\pi_1[i,\infty],\dots,\pi_n[i,\infty]\rangle$.
Every Kripke structure $K$ gives rise to a set of paths from any state
$s\in S$, denoted by
$\Paths_{K,s}$. 
%
To allow the case of a specification referring to atomic
propositions that are not contained in the atomic propositions $\AP$
of the Kripke structure, we assume that these atomic propositions are
not restricted.
Therefore, we define 
\(
\Paths_{K,s}^{\AP'}=\big\{(s_1,P_1\cup P'_1).(s_2,P_2 \cup P'_2).\dots ~\mid~ (s_1,P_1).(s_2,P_2).\dots\in\Paths_{K,s},~ \forall i.~P'_i\subseteq \AP'\setminus\AP\big\}. \)

\subsection{Syntax}

We extend the syntax of CTL* with path variables. Path variables are introduced by path quantifiers. Since formulas may refer to multiple paths at the same time, 
atomic propositions are indexed with the path variable they refer to.
%
%
The formulas of HyperCTL are defined by the following grammar, where
$a\in\AP$ for some set $\AP$ of atomic propositions and $\pi \in
\V$ for an infinite supply $\V$ of path variables:
%
\[
\quad\quad\quad\varphi ~ ::= \quad~ a_{\pi} \quad~ \vert \quad~ \neg\varphi \quad~ \vert \quad~ 
\varphi\vee\varphi
\quad~ \vert \quad~ \LTLcircle\varphi \quad~ \vert \quad~ \varphi \U\varphi 
\quad~ \vert \quad~ \exists \pi.\; \varphi 
\]
A formula is \emph{closed} if all occurrences of some path variable
$\pi$ are in the scope of a path quantifier $\exists \pi$. A HyperCTL
\emph{specification} is a Boolean combination of closed HyperCTL
formulas each beginning with a quantifier (or its negation). A formula
is in \emph{negation normal form} (NNF) if it only contains negations
in front of labels and in front of quantifiers. Every HyperCTL formula
can be transformed into an equivalent formula in negation normal form,
if we introduce the additional operators
$\varphi_1\wedge\varphi_2\equiv \neg(\neg\varphi_1\vee\neg\varphi_2)$
and $\varphi_1\R\varphi_2\equiv \neg(\neg\varphi_1\U\neg\varphi_2)$.

We extend HyperCTL with \emph{universal} path quantification
$\forall\pi.\;\varphi \equiv \neg\exists\pi.\neg\varphi$ as well as
with the usual temporal operators: $\true\equiv a_\pi\vee\neg a_\pi$,
$\false\equiv\neg\true$, $\LTLdiamond\varphi\equiv\true\U\varphi$,
$\LTLsquare\varphi\equiv\neg\LTLdiamond\neg\varphi$,
$\varphi_1\W\varphi_2\equiv\varphi_1\U\varphi_2\vee\LTLsquare\varphi_1$.
For convenience, we also introduce some abbreviations for comparing
paths.  Given a set $P\subseteq \AP$ of atomic propositions, we
abbreviate $\pi[0]\!=_P\!\pi'[0] ~\equiv~ \bigwedge_{a\in P}
a_{\pi}\!\leftrightarrow\! a_{\pi'}$, and $\;\pi\!=_P\!\pi' ~\equiv~
\LTLsquare (\pi[0]\!=_P\!\pi'[0])$, and analogously for $\neq$, e.g.,
$\pi\!\neq_P\!\pi'~\equiv~\neg(\pi\!=_P\!\pi')$.
%
Finally, 
$[\varphi]_\pi$ indicates that the atomic propositions in $\varphi$
are indexed with the path variable $\pi$, and $[\varphi]_{\pi\to
  \pi'}$ is the formula obtained by replacing index $\pi$ by index
$\pi'$ on all (free) atomic propositions in $\varphi$.



%

\subsection{Semantics}
 
Let $K = (S,s_0,\delta,\AP,L)$ be a Kripke structure and $\varphi$ a
HyperCTL formula over atomic propositions $\AP'$.
We define $\Pi\models_K\varphi$, the satisfaction relation of $\varphi$
with respect to a set of paths $\Pi=\tupleof{\pi_1,\dots,\pi_m}$, as follows:
\[
\begin{array}{l@{\hspace{1.5em}}c@{\hspace{1.5em}}l}
  \Pi\models_K a_{\pi_k} & \text{whenever} & a\in \pi_k[0] \\
  \Pi\models_K \neg \psi & \text{whenever} & \Pi\not\models_K\psi \\
  \Pi\models_K \psi_1 \vee \psi_2 & \text{whenever} & \Pi\models_K\psi_1 \text{ or } \Pi\models\psi_2 \\
  \Pi\models_K\LTLcircle\psi & \text{whenever} & \Pi[1,\infty]\models_K\psi  \\
  \Pi\models_K\psi_1\Until\psi_2 & \text{whenever} & \exists i \geq 0.~ \Pi[i,\infty]\models_K\psi_2 \;\;\;\text{and}\;\;\;
\forall 0 \leq j < i.~ \Pi[j,\infty]\models_K\psi_1  \\
   \Pi\models_K\exists\pi.\; \psi & \text{whenever} & \exists \pi\in\Paths_{K,\pi_n[0]}^{\AP'}.~ \Pi+\pi\models_K\psi \;,
\end{array}
\]
%
%
%
where $\Pi+\pi$ is short for $\langle\pi_1,\dots,\pi_n,\pi\rangle$. For the special case that the set of paths is empty, $\Pi=\langle\rangle$, we set
$\Paths_{K,\pi_n[0]}=\Paths_{K,s_0}$.  We say that a Kripke structure
$K$ \emph{satisfies} a HyperCTL formula $\varphi$, denoted as $K\models
\varphi$, if $\langle\rangle \models_K \varphi$ holds true.  The
\emph{model checking problem} for HyperCTL is to decide
whether a given Kripke structure satisfies a given HyperCTL formula.



%
%

\subsection{Examples}
\label{ssect:examples}

HyperCTL has a broad range of applications.  In this subsection, we
review a few motivating examples to illustrate the expressivity of
HyperCTL. Due to space limitations we focus on applications in
security and only sketch a few examples of other areas, such as the
Hamming distance of messages in coding theory.

For simplicity, we consider a synchronous system
model%
\footnote{HyperCTL can also express meaningful properties for
  asynchronous systems. 
} 
for the remainder of this section.  The system reads input variables
$I\subseteq\AP$ and writes output variables $O\subseteq\AP$ at every
step; for simplicity we assume that all variables are binary.  Since
we are not interested in the internal behavior of systems, we specify
our properties only in terms of their input-output behavior.

\paragraph*{Security} 
Informally, most secrecy properties essentially require that secret
inputs shall not influence the observable behavior of a system.  We assume
that the inputs are partitioned into subsets $H$ and $L$, that
indicate the secret and the public variables (often called \emph{high}
and \emph{low} in the literature), and that the observer may only
observe the output variables $O$.  Noninterference attempts to captures
this notion of secrecy, which can be written in HyperCTL, simply
requiring that the output may only depend on the public inputs:
\[
\forall \pi.\forall \pi'.~ (\pi[0]\!=_O\!\pi'[0]) ~\W~ (\pi[0]\!\not=_I\!\pi'[0]).
\]
This notion can be reinterpreted as a requirement of the system to
produce the same output for all paths up to a point in which a
difference in the public input excuses differences in the observable
behavior.  This reinterpretation enables the expression a basic form
of \emph{declassification} by extending the ``excuse'', using the
following HyperCTL pattern:
\[
\forall \pi.\forall \pi'.~ (\pi[0]\!=_O\!\pi'[0]) ~\W~ (\pi[0]\!\not=_I\!\pi'[0] ~\vee~\varphi),
\]
where $\varphi$ would be instantiated by a simple LTL formula (for example on $\pi$) requiring that the user is logged in. This pattern can
express for example that the user may only see data if he is logged
in, and if he is logged out, he may not receive any further updates.
For a more formal discussion of security definitions we refer to
Section~\ref{sect:ExamplesSecurity}.

\paragraph*{Coding Theory}

Following the seminal work of Shannon, many encoding schemes have been
designed to provide error-resistant transmission of data.  However,
the formal verification of functional correctness of implementations
has received little attention so far.  HyperCTL can easily express
important functional correctness criteria for encoders and decoders.

We first define the fundamental notion of Hamming distance.  The
Hamming distance is the number of positions of two sequences of the
same length that differ.  The formula $\text{Ham}_{P}$ expresses that
two sequences of states have a Hamming distance of at most $d$, considering only differences in the propositions $P\subseteq\AP$:
\[
\renewcommand{\arraystretch}{1.2}
\begin{array}{rl}
\text{Ham}_{P}(0,\pi,\pi')&\hspace{-5pt}=~ \pi\!=_P\!\pi'\;, \\
\text{Ham}_{P}(d,\pi,\pi')&\hspace{-5pt}=~ \pi[0]\!=_P\!\pi'[0] ~\W~ \big(\pi[0]\!\not=_P\!\pi'[0] ~\wedge~ \LTLcircle (\text{Ham}_{P}(d\!-\!1,\pi,\pi'))\big) 
\end{array} 
\]
The Hamming distance is often used as a metric for the error
resistance of a code.  If all paths of a code have a minimal Hamming
distance of $d$, the decoder can retrieve the original word even if up
to $\lfloor\frac{d-1}{2}\rfloor$ errors happened.  

For the scenario in which the encoder transforms an input stream, represented by propositions $I$ into codewords that are represented by propositions $C$, we can specify that an encoder produces only code
words that have minimal Hamming distance $d$:
\[
\forall\pi.\forall\pi'.~ ~\pi\!\not=_I\!\pi' \to \neg\text{Ham}_{C}(d-1,\pi,\pi')~. 
\]
%
Similarly, we could formulate a correctness condition for decoders.
For every code word, and every message that differs from that code
word in less than $\lfloor\frac{d-1}{2}\rfloor$ positions, the decoder
must produce the same decoded message.  Note that these formulations
are completely independent of parameters like the block length and thus can be applied to compare different kinds of coding schemes.

\paragraph*{Edit distance}
Similarly to the Hamming distance, we can define other metrics based
on the insertion or deletion of elements.  The following formula
fragment requires that paths $\pi$ and $\pi'$ differ only in that
$\pi'$ has an additional position:
\[
\pi[0]\!=_P\!\pi'[0] ~\W~ \big(\pi[0]\!\not=_P\!\LTLcircle\pi'[0] \big)~.
\]
We use this observation in Section~\ref{sect:ExamplesSecurity} to
encode security properties for non-synchronous systems.

\section{Related logics}
\label{sect:relations}



\subsection{Linear- and Branching-Time Temporal Logics}

HyperCTL is an extension of CTL*, and therefore
subsumes the usual temporal logics LTL, CTL, and
CTL*~\cite{Emerson+Halpern/1986/BranchingVersusLinearTimeTemporalLogic,Kupferman+Vardi+Wolper/2000/AnAutomataTheoreticApproachToBranchingTimeMC}.

\begin{theorem}
HyperCTL subsumes LTL, CTL, and CTL*.
\end{theorem}

A more interesting case is quantified propositional temporal logic
(QPTL)~\cite{Sistla+Vardi+Wolper/1987/TheComplementationProblemForBuchiAutomata}.
QPTL also extends LTL with quantification, but HyperCTL quantifies
over paths and QPTL over propositions. In terms of expressiveness, we
prove now that HyperCTL strictly subsumes QPTL.  Informally,
quantification over paths is more powerful than quantification over
propositions, because the set of paths depends on the Kripke
structure, while the truth values of the quantified propositions do
not.  We will also exploit in Section~\ref{sect:MCandSat} the relation
between HyperCTL and QPTL to obtain a model checking algorithm for
HyperCTL.
QPTL formulas are generated by the following grammar,
where $p\in\AP$:
\[
\psi ~~ ::= ~~ p ~~\vert~~ \neg\psi ~~\vert~~
\psi\vee\psi
~~\vert~~ \LTLcircle\psi ~~\vert~~ \LTLdiamond\psi
~~\vert~~ \exists p.\; \psi ~,
\] 
QPTL formulas are interpreted over paths with all operators inheriting
the LTL semantics except $\exists\pi.\psi$:
\[
\pi\models\exists p.\;\psi\;\;\;\;\text{ whenever }\;\;\;\; \exists \pi'\in (2^{\AP})^{\omega}.~\pi=_{\AP\setminus p}\pi' ~\wedge~ \pi'\models\psi\;.
\]
A Kripke structure satisfies a QPTL formula if all paths from the
initial state satisfy the formula.

\begin{theorem}
HyperCTL subsumes QPTL.
\end{theorem}

\begin{proof}(\textit{sketch}) In order to express a QPTL formula in
  HyperCTL, we rename all bound propositions with unique fresh names
  $\AP'$. We use these propositions as free variables, which are
  unconstrained because they do not occur in the Kripke structure.
  Then, each propositional quantification $\exists p$ in the QPTL
  formula is replaced by a path quantification $\exists \pi_p$ in the
  HyperCTL formula, and each occurrence of $p$ is replaced by
  $p_{\pi_p}$.
\end{proof}

\begin{theorem}
QPTL does not subsume HyperCTL.
\end{theorem}

\begin{proof}(\textit{sketch})
  Since QPTL formulas express path properties, which are checked on
  all paths, QPTL cannot express properties that express the
  \emph{existence} of paths such as $\exists \pi . \LTLcircle p_\pi$.
\end{proof}

\paragraph*{EQCTL* and QCTL*} 

The logic
EQCTL*~\cite{Kupferman/1995/EQCTL,Kupferman+Vardi+Wolper/2000/AnAutomataTheoreticApproachToBranchingTimeMC}
extends CTL* by the \emph{existential} quantification over
propositions.
QCTL*~\cite{French/2001/DecidabilityOfQuantifiedPropositionalBranchingTimeLogics}
additionally allows for negated existential (i.e. universal)
quantifiers.  However, in contrast with our work, EQCTL* and QCTL* do
not unify the quantification over paths and the quantification over
propositions.  Similar to the result on satisfiability of HyperCTL,
the satisfiability of QCTL* is highly undecidable, which is shown
using a different proof technique.  To the best of the authors'
knowledge, no model checking algorithm has been proposed for QCTL*.
It is straightforward to derive one from the work in this paper.  A
restricted version of QCTL*, building on CTL instead of CTL*, allows
for a comparably low worst-case complexity of the model checking
problem~\cite{Patthak+al/2002/QCTL}.  It is open whether this result
can be transferred to HyperCTL.

\subsection{Epistemic Logics}

Epistemic
logics~\cite{Fagin+Halpern+Moses+Vardi/1995/ReasoningAboutKnowledgeBook}
reason about the \emph{knowledge} of agents and have gained increasing
popularity as specification languages for
security
over the past years~
\cite{Balliu+Dam+Guernic/2011/EpistemicTemporalLogicForIFSecurity,Chadha+al/2009/EpistemicLogicsAppliedPiCalculus,vanderMeyden+Wilke/2007/PreservationOfEpistemicPropertiesInSecurityProtocolImplementations}. 
Since the semantics of epistemic logics relates
multiple possible worlds (paths in the setting of temporal
properties), these logics define hyperproperties.  We now study
the relation between epistemic logics and HyperCTL.

\emph{Epistemic temporal logic} with perfect recall
semantics~\cite{Meyden/1993/AxiomsForKnowledgeAndTimeInDistributedSystemsWithPerfectRecall,Balliu+Dam+Guernic/2011/EpistemicTemporalLogicForIFSecurity,Fagin+Halpern+Moses+Vardi/1995/ReasoningAboutKnowledgeBook}
extends LTL with the \emph{knowledge} operator ${\sf K}_i\psi$, which
specifies that an observer $i$ may \emph{know} $\psi$. This logic is
defined on \emph{interpreted systems}. An interpreted system
$I=(R,L,Q,\AP)$ for $m$ agents (or processors) consists of a set of
possible observations $Q$, a set of runs $R\subseteq(Q^m)^\omega$, and
a labeling function $L:R\times\mathbb{N}\to 2^\AP$.  In general, each
agent can only observe \emph{changes} to its current observation.
Let $r|_n$ be the prefix of a run, up to and including step $n$, and
let $r|_{n,i}$ be the prefix of a run projected to the states of an
agent $i$.  The knowledge of an agent then corresponds to
$\trace(r|_{n,i})$, which is defined for $n=1$ as $\trace(q_i)=q_i$, and otherwise as 
$\trace(r|_{n-2,i}.q_i.q'_i)=\trace(r|_{n-2,i}.q_i)$ if
$q_i=q'_i$ and 
$\trace(r|_{n-2}.q_i.q'_i)=\trace(r|_{n-2,i}.q_i).q'_i$ if
$q_i\not=q'_i$.

Under the \emph{synchronous time assumption}~\cite{Fagin+Halpern+Moses+Vardi/1995/ReasoningAboutKnowledgeBook},
agents may observe the occurrence of every step.  
We capture this assumption by requiring that the observations of all agents include the clock, e.g. $Q=Q'\times\mathbb{N}$, and let the runs provide every agent with the correct step number. 
We then have that $\trace(r|_{n,i})=r|_{n,i}$.

For a given run $r\in R$ and a step number $n$, the satisfaction relation for the knowledge operator is defined as follows: 
\[
\begin{array}{l}
r,n\models{\sf K}_i\psi \quad\text{whenever}\quad\forall r'\in R.~ \trace(r|_{n,i})\!=\!\trace(r'|_{n,i}) ~\to~ r',n\!\models\psi \;
\end{array}
\]

An interpreted system $\I=(R,L,Q,\AP)$ induces a unique Kripke
structure $K(\I)=(S,s_0,\delta,\AP',L')$ with
$S=Q^m\times\mathbb{N}\cup\{s_0\}$, where $s_0$ is a fresh initial
state, and the runs $R$ define the transition relation $\delta$.
For the initial state $s_0$, $\delta(s_0)$ leads to the all initial states of
the interpreted system, which are defined by all prefixes of length 1
of some run in $R$. The new set of propositions $\AP'=\AP\times Q^m$
includes local observations for all agents, to later encode their
knowledge in HyperCTL.  The resulting labeling function is unique.

Note that for finite interpreted systems, that is if $Q$ and $\AP$ are
finite, $L$ is only dependent on the current states and the runs $R$
are generated by a relation on $Q^m$. Hence, the resulting Kripke
structure is finite as well.

Similarly, we can define the semantics of the knowledge operator on
Kripke structures.
We have that $\pi,i \models {\sf K}_P \varphi$ iff
\[
\begin{array}{l}
\forall\pi'\in\Paths_{K}.~ \trace(\pi[0,i],P)\!=_P\!\trace(\pi'[0,i],P) ~\to~ \pi',i\!\models\varphi \;,
\end{array}
\]
where $\trace(\pi,P)$ (the type being $\trace:S^*\times 2^\AP\to S^*$)
is the projection to the points in $\pi$ where propositions in $P$
change.  For example, $\trace(s_1.s_2.s_3.s_4.s_4)=s_1.s_3$, assuming
$L(s_1)=\{a\}$, $L(s_2)=\{a,c\}$, $L(s_3)=\{a,b,c\}$,
$L(s_4)=\{a,b\}$, and $P=\{a,b\}$.



The following theorem shows that, under
the synchronous time assumption, HyperCTL subsumes
epistemic temporal logic.


\begin{theorem}\label{thm:synchEpistemic}
  For every epistemic temporal logic formula $\psi$ with synchronous
  time assumption and interpreted system $\I$, there is a HyperCTL
  formula $\varphi$ such that $\I\models\psi$ iff
  $K(\I)\models\varphi$.
\end{theorem}

%
%

\begin{proof}
  We start by considering extension HyperCTL with the \emph{knowledge}
  operator with the semantics above, and prepend a universal path
  quantifier. Then we apply a step-wise transformation to eliminate
  all knowledge operators.
	
  Let $\varphi$ be a HyperCTL formula in NNF and in prenex normal form
  (note that epistemic logic formulas prepended with a universal path
  quantifier are in prenex normal form).  For brevity we summarize the
  leading quantifiers of $\varphi$ with $\mathbf{Q}$, such that
  $\varphi=\mathbf{Q}.\varphi'$ where $\varphi'$ is quantifier free.
  Let $t$ and $u$ be propositions that $\varphi$ does not refer to.
	
	For the case that a knowledge operator ${\sf K}_P\psi$ occurs with positive polarity, we translate $\varphi$ into the following 
	HyperCTL formula: 
\[
\begin{array}{l}
\mathbf{Q}. \exists\pi. ~\varphi'|_{{\sf K}_P\psi\to u_{\pi}}\;\wedge\;
\big(
\forall\pi'.~ (t_{\pi'}\U\LTLsquare \neg t_{\pi'}) ~\to~  \\\qquad  \forall\pi''.~   \LTLsquare(t_{\pi'}\to (\pi_n[0]\!=_P\!\pi''[0])) ~\to~ \LTLsquare(t_{\pi'}\wedge u_{\pi}\to [\psi]_{\pi''})\big)\;,
\end{array}
\]
and, if the knowledge operator occurs negatively,
\[
\begin{array}{l}
\mathbf{Q}. \exists\pi. ~\varphi'|_{\neg{\sf K}_P\psi\to u_{\pi}}\;\wedge\;
\big(\forall\pi'.~ (t_{\pi'}\U\LTLsquare \neg t_{\pi'}) \to \\\qquad  \exists\pi''.~ \LTLsquare(t_{\pi'}\to (\pi_n[0]\!=_P\!\pi''[0]))  ~\wedge~ \LTLdiamond(t_{\pi'}\!\wedge u_{\pi}\wedge \neg [\psi]_{\pi''})\big)\;,
\end{array}
\]
where $\varphi'|_{{\sf K}_P\psi\to u_{\pi}}$ denotes that in
$\varphi'$ one positive or one negative occurrence, respectively, of the knowledge
operator ${\sf K}_P\psi$ gets replaced by proposition $u$.  We repeat this transformation until no knowledge
operators remain, using the fact that every HyperCTL formula can be
transformed into prenex normal form.

The intuition of this transformation is to label the path (more
precisely the tuple of paths $\Pi$ that is determined by the
quantifiers $\mathbf Q$) with binary signals $t$ and $u$ that indicate
when to check that $\phi$ has to hold on all (or some) alternative
paths.  By replacing the knowledge operator by the existentially
quantified proposition $u$,
we select a satisfying interpretation (if there is any) among the
disjunctions concerning the application of the knowledge operator.
Hence, proposition $u_{\pi}$ at position $i$ indicates whether to
check that all paths that are equal to $\pi_n$ up to position $i$ must
satisfy $\Pi,i\models_K\psi$.  We express that by universally
selecting a position, indicated by the point at which $t_{\pi'}$
switches from true to false (up to that point the paths must be
equal) and then requiring that when we reach that position and the
flag $u_{\pi}$ is set, we also have to check $\psi$ on~path~$\varphi$.

The negated case follows the same intuition on the signals $u$ and
$t$, except that the same (universal) quantifier cannot be reused to
select the alternative path and the signal $t$.
\end{proof}

Without the synchronous time assumption, epistemic temporal logic is still subsumed by HyperCTL, if we add \emph{stutter steps} to the modified Kripke structure; let $K'(\I)$ denote the adapted Kripke structure.  
We give details on the modified transformation and the encoding of the formula in Appendix~\ref{app:asynchEpistemic}. 

\newcounter{thm-asynchEpistemic}
\setcounter{thm-asynchEpistemic}{\value{theorem}}
\begin{theorem}\label{thm:asynchEpistemic}
For every epistemic temporal logic formula $\psi$ and interpreted system $\I$, there  is a HyperCTL formula $\varphi$ such that 
$\I\models\psi$ iff $K'(\I)\models\varphi$. 
\end{theorem}

\subsection{SecLTL}
\label{ssect:SecLTL}

SecLTL~\cite{Dimitrova+Finkbeiner+Kovacs+Rabe+Seidl/12/SecLTL} extends
temporal logic with the \emph{hide} modality $\Hide$, which allows us to express
information flow properties such as
noninterference~\cite{Goguen+Meseguer/1982/SecurityPoliciesAndSecurityModels}.
%
%
The semantics of SecLTL
is
defined in terms of labeled transition system, where the edges are
labeled with evaluations of the set of variables.  
The formula $\Hide_{H,O}\varphi$
specifies that the current evaluations of a subset $H$ of the
input variables $I$ are kept secret from an attacker who may observe
variables in $O$ until the release condition $\varphi$ becomes true.
The semantics is formalized in terms of a set of \emph{alternative paths}
to which the \emph{main path} is compared: 
\[
\AltPaths(\pi,H) = \{\pi'\in\Paths_{K,\pi[0]} \mid \pi[1]\!=_{I\setminus H}\!\pi'[1]\wedge \pi[2,\infty]\!=_{I}\!\pi'[2,\infty]\}
\]
A path $\pi$ satisfies the SecLTL formula $\Hide_{H,O}\varphi$, denoted by $\pi\models \Hide_{H,O}\varphi$, iff
\[
\begin{array}{l}
\forall\pi'\in\AltPaths(\pi,H).~ \big(\pi\!=_{O}\!\pi'\;
%
%
\;\vee\;
\exists i\geq0. ~(\pi[i,\infty]\models_K\varphi) \;\wedge\; \pi[1,i\!-\!1]\!=_{O}\!\pi'[1,i\!-\!1]\big) \;.
\end{array}
\]
A transition system $M$ satisfies a SecLTL formula $\psi$, denoted by $M \models \psi$, if every path $\pi$ starting in the initial state satisfies $\psi$.
To encode the hide modality in HyperCTL, we
first translate $M$ into a Kripke structure $K(M)$ whose states
are labeled with the evaluation of the variables on the edge leading
into the state.
The initial state is labeled with the empty set.  In the modified
system, $L(\pi[1])$ corresponds to the current labels.  
%
We encode $\Hide_{H,O}\varphi$ as the following HyperCTL formula: 
\[
\forall\pi'.\;\pi[1]\!=_{I\setminus H}\!\pi'[1] \wedge \LTLcircle\big(\pi[1]\!=_{O}\!\pi'[1]\;\W\;(\pi[1]\!\not=_{I}\!\pi'[1]\vee \varphi)\big) 
\]


\begin{theorem}
For every SecLTL formula $\psi$ and transition system $M$, there  is a HyperCTL formula $\varphi$ such that 
$M\models\psi$ iff $K(M)\models\varphi$. 
\end{theorem}

The model checking problem for SecLTL is PSPACE-hard in the size of
the Kripke
structure~\cite{Dimitrova+Finkbeiner+Kovacs+Rabe+Seidl/12/SecLTL}.
The encoding of Sec\-LTL specifications in HyperCTL implies that the
model checking problem for HyperCTL is also PSPACE-hard (for
a fixed specification of alternation depth $\geq 2$), as claimed in
Theorem~\ref{thm:PSPACElowerBound}.


\section{Model Checking and Satisfiability}
\label{sect:MCandSat}

\subsection{Model Checking}\label{ssect:modelchecking}

In this section we exploit the connection between HyperCTL and QPTL to
obtain a model checking algorithm for HyperCTL. We reduce the model
checking problem for HyperCTL to the satisfiability problem for QPTL.
First, we encode the Kripke structure as a QPTL formula. Then we
qualify all quantifiers to refer to the interpretations of the
propositions that occur on the corresponding path in the Kripke
structure.  The satisfiability problem for the resulting QPTL formula is nonelementary in the quantifier
alternation depth.

\vspace{4pt}\begin{definition}[Alternation Depth] A HyperCTL or QPTL
  formula $\varphi$ in NNF has alternation depth 1 plus the highest
  number of alternations from existential to universal and universal
  to existential quantifiers along any of the paths of the formula's
  syntax tree.  Occurrences of $\Until$ and $\Release$ count as an
  additional alternation.
\end{definition}

\begin{theorem}
\label{thm:nonelementaryformula}
The model checking problem for HyperCTL specifications $\varphi$ with
alternation depth $k$ is complete for NSPACE$(g(k-1,|\varphi|))$.
\end{theorem}

\begin{proof}
The satisfiability problem for QPTL formulas $\varphi$ in prenex
  normal form and with alternation depth $k$ is complete for
  NSPACE$(g(k-1,|\varphi|))$~\cite{Sistla+Vardi+Wolper/1987/TheComplementationProblemForBuchiAutomata}.

For the upper bound on the HyperCTL model checking complexity, we encode Kripke structures as QPTL formulas (see e.g.~\cite{Kesten+Pnueli/2002/CompleteProofSystemForQPTL,Manna+Pnueli/1992/TheTemporalLogicOfReactiveSystems}) and translate until operators $\psi \;\mc U\; \psi'$~as
\[
\exists t. ~t\hspace{2pt}\wedge\hspace{2pt}\LTLsquare(t\hspace{2pt}\to\hspace{2pt} \psi'\vee(\psi\wedge\LTLcircle t)) \hspace{2pt}\wedge\hspace{2pt} \neg\hspace{2pt}\LTLsquare \hspace{2pt}t\;.
\]
HyperCTL path quantifiers are encoded by extending the set of atomic
propositions, introducing for each quantifier an additional copy of
$\AP$.  That is, assuming that all paths have a unique name, we
translate $\exists\pi_k.~\psi$ to $\exists \AP_k. [\psi]_{\pi_k\to
  \AP_k}$, where $[\psi]_{\pi_k\to \AP_k}$ means that we let the
atomic propositions $a_{\pi_k}$ in $\psi$ refer to the proposition $a$
in the $k$-th copy of $\AP$.

For the lower bound, we reduce the satisfiability problem for a given
QPTL formula $\varphi$ in prenex normal form to a model checking
problem $K \models \varphi'$ for some Kripke structure $K$ and some
HyperCTL formula $\varphi'$. We assume, without loss of generality,
that $\varphi$ is closed (if a free proposition occurs in $\varphi$,
we bind it with an existential quantifier) and each quantifier in
$\varphi$ introduces a different proposition.

The Kripke structure $K = (S,s_0,\delta,\AP,L)$ 
consists of two states $S=\{s_0, s_1\}$, is fully connected $\delta: s \mapsto S$ for $s \in S$, and is labeled with the truth values of a single atomic proposition $\AP=\{p\}$, $L: s_0 \mapsto \emptyset, s_1 \mapsto \{p\}$.   

The QPTL quantifiers choose a valuation of the quantified propositions
already in the first position of a sequence while the path quantifiers
of HyperCTL can only pick a new state in the \emph{next} position.
Hence, we shift all atomic propositions in $\varphi$ by one position,
by replacing every atomic proposition $q$ in $\varphi$ by $\LTLcircle
q$, resulting in the new QPTL formula $\psi$.
The final HyperCTL formula $\varphi' = T_1(\psi)$ is constructed recursively from $\psi$ as follows:
\[
\mc T_m (\psi) =
\begin{cases}
  \exists \pi_{m}.~ \mc T_{m+1}(\psi'[t \mapsto p_{\pi_{m}}]) &
  \text{if }\psi=\exists t.\psi'\\
  \neg\exists \pi_{m}.~ \mc T_{m+1}(\psi'[t \mapsto p_{\pi_{m}}]) & 
  \text{if }\psi=\neg\exists t.\psi'\\
  \psi &  \text{otherwise}
\end{cases}
\]
\end{proof}

An NLOGSPACE lower bound in the size of the Kripke structure for fixed
specifications with alternation depth 1 follows from the non-emptiness
problem of non-deterministic B\"uchi automata.  For alternation depth
2 and more we can derive PSPACE hardness in the size of the Kripke
structure from the encoding of the logic SecLTL into HyperCTL (see
Subsection~\ref{ssect:SecLTL}).

\begin{theorem}
\label{thm:PSPACElowerBound}
For HyperCTL formulas the model checking problem is hard for PSPACE in
the size of the system.
\end{theorem}


\paragraph*{A Remark on Efficiency}
The use of the standard encoding of the until operator in QPTL with an additional quantifier is, in certain cases, wasteful. The satisfiability of QPTL formulas can be checked with an automata-theoretic construction, where we first transform the formula into prenex normal form,
then generate a nondeterministic B\"uchi automaton for the
quantifier-free part of the formula, and finally apply projection and
complementation to handle the existential and universal quantifiers.
In this way, each quantifier alternation, including the alternation introduced by the encoding of the until operators, causes an exponential blow-up. 
However, if an until operator occurs in the
quantifier-free part, the standard transformation of LTL formulas to
nondeterministic B\"uchi automata could handle this operator
without quantifier elimination, which would result in an exponential speedup.

Using this insight, the model checking complexity for many of the formulas presented above and in
Section~\ref{sect:ExamplesSecurity} could be reduced by one
exponent.  Additionally, the
complexity with respect to the size of the system would reduce to NLOGSPACE
for HyperCTL formulas where the leading quantifiers are all of the same type
and are followed by some quantifier-free formula which may contain until
operators without restriction.

\subsection{Satisfiability}

Unfortunately, the positive results regarding the decidability of the model
checking problem for HyperCTL do not carry over to the satisfiability
problem.  The \emph{finite-state satisfiability problem} consists of
the existence of a finite model, while the \emph{general
  satisfiability problem} asks for the existence of a possibly
infinite model. The following theorem shows that both versions are
undecidable.

\begin{theorem}
  For HyperCTL, finite-state satisfiability is hard for $\Sigma^0_1$ and
  general satisfiability is hard for $\Sigma^1_1$.
\end{theorem}

\begin{proof}
  We give a reduction from the synthesis problem for LTL
  specifications in a distributed architecture consisting of two
  processes with disjoint sets of variables.  The synthesis problem
  consists on deciding whether there exist transition systems for the
  two processes with input variables $I_1$ and $I_2$, respectively,
  and output variables $O_1$ and $O_2$, respectively, such that the
  synchronous product of the two transition systems satisfies a given
  LTL formula $\varphi$. This problem is hard for $\Sigma^0_1$ if the
  transition systems are required to be finite, and hard for
  $\Sigma^1_1$ if infinite transition systems are allowed~(Theorems
  5.1.8 and 5.1.11 in \cite{Rosner/1992/Thesis}).

  To reduce the synthesis problem to HyperCTL satisfiability, we
  construct a HyperCTL formula $\psi$ as a conjunction $\psi=\psi_1
  \wedge \psi_2 \wedge \psi_3$.  The first conjunct ensures that
  $\varphi$ holds on all paths: $\psi_1 = \forall \pi. [\varphi]_\pi$.
  The second conjunct ensures that every state of the model has a
  successor for every possible input: $\forall \pi . \LTLsquare
  \bigwedge_{I \subseteq I_1 \cup I_2} \exists \pi' \LTLcircle
  \bigwedge_{i \in I} i \ \bigwedge_{i \not\in I} \neg i$. The third
  conjunct ensures that the output in $O_1$ does not depend on $I_2$
  and the output in $O_2$ does not depend on $I_1$:
  $\psi_3=\forall\pi.\forall\pi'.~\big(\pi\!=_{I_1}\!\pi' \to
  \pi\!=_{O_1}\!\pi'\big)\wedge\big(\pi\!=_{I_2}\!\pi' \to
  \pi\!=_{O_2}\!\pi'\big)$.

The distributed synthesis problem has a (finite) solution iff the HyperCTL formula $\psi$ has a (finite) model.
\end{proof}

\section{Applications in Security}\label{sect:ExamplesSecurity}

Information flow is a much studied field in security~
\cite{Gray/1991/TowardAMathematicalFoundationForIFSecurity,Barthe+DArgenio+Rezk/04/SecureInformationFlowBySelfComposition,Terauchi+Aiken/05/SecureInformationFlowAsSafetyProblem,Yasuoka+Terauchi/2012/QuantitativeInformationFlowAsSafetyAndLivenessHyperproperties,Goguen+Meseguer/1982/SecurityPoliciesAndSecurityModels,Sabelfeld+Sands/2005/DimensionsAndPrinciplesOfDeclassification,Mantel/2000/PossibilisticDefinitionsofSecurityAnAssemblyKit,Banerjee+Naumann+Rosenberg/2008/ExpressiveDeclassificationPoliciesAndModularStaticEnforcement,Askarov+Myers/2010/ASemanticFrameworkForDeclassificationAndEndorsement,Mclean/1992/ProvingNoninterferenceAndFunctionalCorrectnessUsingTraces,Amtoft+Banerjee/2004/InformationFlowAnalysisInLogicalForm,Koepf+Basin/2007/AnInformationTheoreticModelForAdaptiveSideChannelAttacks,Clarkson+Schneider/10/Hyperproperties,Balliu+Dam+Guernic/2011/EpistemicTemporalLogicForIFSecurity,Zdancewic+Myers/03/ObservationalDeterminism}. 
%
The question of interest is whether a system reveals information about a
cryptographic key or personal data, generally called \emph{the
  secret}, to an observer or attacker.  This problem goes beyond
classical data flow analysis, in the sense that not only forbids
the system from revealing the secret as one of its outputs, but it
also requires that the system's observable behavior is independent
from the secret.

For example, consider the case of a hash set that stores both public
and secret elements. If an attacker observes the list of public
elements in the set he may learn partial information about the
presence of secret elements and their values by observing the
ordering of the low elements.
%
%
Even though in this case a security attack is not immediate and it
might only work under additional assumptions, secrets in the data
structure may be compromised.  Practical attacks based on such partial
information have been implemented. For example, the well known attack
on SSL \cite{Brumley+Boneh/2003/RemoteTimingAttacksArePractical} uses
subtle timing differences in the answers from a server.

The notion of \emph{noninterference} by Goguen and Meseguer is the
starting point of most of the bast body of work on information flow.
This notion states the secrecy requirement in terms of two groups of
users, which represent the secret (their behavior must be secret) and
the attacker, respectively.  The group of attackers may not observe a
different behavior of the system, independent of whether the secret
users perform actions or not.

\begin{theorem}\label{thm:noninterference}
  There is a HyperCTL formula $\varphi_{\mathit{NI}}(G_H,G_L)$ and an
  encoding $K(M)$ of state machines into Kripke structures such that
  for every state machine $M$, and groups of users $G_H$ and $G_L$,
  $G_H$ does not interfere with $G_L$ in $M$ iff
  $K(M)\models\varphi_{\mathit{NI}}(G_H,G_L)$ holds.
\end{theorem}

The encoding of state machines into Kripke structures and the proof are straightforward and can be found in Appendix~\ref{app:noninterference}. 

A multitude of variations of the notion of secrecy have been
proposed~
\cite{Sabelfeld+Sands/2005/DimensionsAndPrinciplesOfDeclassification,Mantel/2000/PossibilisticDefinitionsofSecurityAnAssemblyKit,Banerjee+Naumann+Rosenberg/2008/ExpressiveDeclassificationPoliciesAndModularStaticEnforcement,Askarov+Myers/2010/ASemanticFrameworkForDeclassificationAndEndorsement,Mclean/1992/ProvingNoninterferenceAndFunctionalCorrectnessUsingTraces,Amtoft+Banerjee/2004/InformationFlowAnalysisInLogicalForm,Koepf+Basin/2007/AnInformationTheoreticModelForAdaptiveSideChannelAttacks,Zdancewic+Myers/03/ObservationalDeterminism}.
%
Even though
hyperproperties~\cite{Clarkson+Schneider/10/Hyperproperties} provide a
useful semantical framework to capture a broad range of these
information flow
properties,~\cite{Clarkson+Schneider/10/Hyperproperties} does not
provide any algorithmic approaches.
The technique of
self-composition~\cite{Barthe+DArgenio+Rezk/04/SecureInformationFlowBySelfComposition}
can be used to cast the verification of these information flow
properties in terms of safety (and sometimes in terms of liveness)
problems
\cite{Terauchi+Aiken/05/SecureInformationFlowAsSafetyProblem,Yasuoka+Terauchi/2012/QuantitativeInformationFlowAsSafetyAndLivenessHyperproperties}
by \emph{modifying the system}.  However, these encodings can be
rather involved (cf.~\cite{HuismanWS/06/TLCharacterisationOfOD}) and
hence, for many of the information flow variations, no algorithms
exist yet that allow for their precise verification.
%

In contrast, HyperCTL allows us to easily derive such algorithms for
many of the proposed notions of information flow without the need to
modify the system.  In the rest of this section, we illustrate the
expressiveness of HyperCTL by showing how to define non-interference
and also more modern definitions of security\footnote{The purpose of
  this section is not to argue for or against particular information
  flow policies.}.
From the results of Section~\ref{ssect:modelchecking} it follows that
the formulas presented in this section yield efficient model checking
algorithms, particularly considering the remark on efficiency.  These
examples suggest that our common logical framework of HyperCTL may
also lead to advances for the practical verification of security
policies.

\subsection{Security Definitions for Nondeterministic Systems}

A common concern of many pieces of work following the seminal work of
Goguen and Meseguer is the extension of noninterference to
nondeterministic systems
(e.g.~\cite{Zakinthinos+Lee/1997/AGeneralTheoryOfSecurityProperties,Mantel/2000/PossibilisticDefinitionsofSecurityAnAssemblyKit,Zdancewic+Myers/03/ObservationalDeterminism,Mclean/1992/ProvingNoninterferenceAndFunctionalCorrectnessUsingTraces}).
%
%
Different approaches differ in their assumptions on the source of
nondeterminism and discuss what requirements would make a
nondeterministic system intuitively secure.

\emph{Observational determinism}
\cite{Zdancewic+Myers/03/ObservationalDeterminism,Mclean/1992/ProvingNoninterferenceAndFunctionalCorrectnessUsingTraces}
is a conservative approach to nondeterminism that assumes that the
attacker does no only know the model, but can choose an arbitrary
implementation (refinement) thereof.  Essentially, the property requires that for a given input
all possible executions of the system look identical for an observer
independently of the choice of the secret and of how the
nondeterminism is resolved.  Huisman et
al.~\cite{HuismanWS/06/TLCharacterisationOfOD} encoded this
information flow property as a CTL* formula after transforming 
the system via
self-composition~\cite{Barthe+DArgenio+Rezk/04/SecureInformationFlowBySelfComposition}.
HyperCTL can naturally express properties over multiple paths
and thus the property can be expressed without the need for
modifications to the system.  See Appendix~\ref{app:OD} for details.


\newcounter{thm-obsDet}
\newcounter{aux}
\setcounter{thm-obsDet}{\value{theorem}}

\begin{theorem}\label{thm:observationalDeterminism}
  There is an encoding of programs as defined in 
  \cite{HuismanWS/06/TLCharacterisationOfOD} into Kripke
  structures, such that HyperCTL can express observational determinism. 
\end{theorem}

The proof shows that self-composition can be replaced by quantifying
over two paths.


\subsection{Security Lattices} 

It is common in security to allow for multiple hierarchically ordered
security levels that are (statically) assigned to users and/or data.
HyperCTL can express these concerns by a conjunction of multiple
secrecy statements that separate the security levels.  The formulation
of security lattices is independent of the precise notion of secrecy
chosen.

\subsection{Declassification}

Many security critical systems reveal secret information partially or
under certain conditions as their normal mode of operation.  For
example, a piece of software that checks passwords must reveal whether
the entered password is correct or not.  Noninterference, though,
categorically rules out such behavior because it leaks some
information to an attacker. Thus, weaker notions of security have been
proposed~\cite{Sabelfeld+Myers/2004/AModelForDelimitedRelease,Banerjee+Naumann+Rosenberg/2008/ExpressiveDeclassificationPoliciesAndModularStaticEnforcement,Askarov+Myers/2010/ASemanticFrameworkForDeclassificationAndEndorsement,Dimitrova+Finkbeiner+Kovacs+Rabe+Seidl/12/SecLTL}
(see~\cite{Sabelfeld+Sands/2005/DimensionsAndPrinciplesOfDeclassification}
for more literature).  There is a variety of patterns to integrate
exception into information flow policies, which are referred to as the
\emph{dimensions} of
declassification~\cite{Sabelfeld+Sands/2005/DimensionsAndPrinciplesOfDeclassification}.

The example of the password checking could be generalized to the
pattern (or dimension) that a predefined fact about the secret may be
released,
e.g.~\cite{Sabelfeld+Myers/2004/AModelForDelimitedRelease}.
Other patterns are that after a certain condition has been fulfilled
information may be released (for example, after access has been
granted)~\cite{Dimitrova+Finkbeiner+Kovacs+Rabe+Seidl/12/SecLTL,Broberg+Sands/2006/FlowLocksTowardsACoreCalculusForDynamicFlowPolicies},
or that some input must be considered secret only under certain
conditions (for example, the keyboard input is only secret while the cursor is
in a password field)~\cite{Dimitrova+Finkbeiner+Kovacs+Rabe+Seidl/12/SecLTL}.



\begin{theorem}
HyperCTL subsumes the declassification properties encoded in \cite{Balliu+Dam+Guernic/2011/EpistemicTemporalLogicForIFSecurity}. 
\end{theorem}

The result follows from the encoding~\cite{Balliu+Dam+Guernic/2011/EpistemicTemporalLogicForIFSecurity} of declassification
properties
into epistemic temporal
logic~\cite{Fagin+Halpern+Moses+Vardi/1995/ReasoningAboutKnowledgeBook}
and the fact that HyperCTL subsumes epistemic temporal logic
(see Section~\ref{sect:relations}).
Appendix~\ref{app:declassification} contains a discussion of the
example of a password checker above to show that declassification
properties can indeed be expressed elegantly in HyperCTL.


\subsection{Quantitative Information Flow}

In cases where leaks cannot be prevented absolutely, we might want to limit the \emph{quantity} of the lost information~\cite{Gray/1991/TowardAMathematicalFoundationForIFSecurity,Clark+Hunt+Malacaria/2002/QuantitativeAnalysisOfTheLeakageOfConfidentialData,Koepf+Basin/2007/AnInformationTheoreticModelForAdaptiveSideChannelAttacks,Smith/2009/OnTheFoundationsOfQantitativeInformationFlow,Yasuoka+Terauchi/2010/OnBoundingProblemsOfQuantitativeInformationFlow}. 
First of all, this quantification requires to measure information by some definition of \emph{entropy}. It turns out that for different practical scenarios different notions of this measure are appropriate. 
To define meaningful measures of entropy of a secret, secrets are generally assumed to be chosen randomly and independently and the system under consideration is assumed to be deterministic (not nondeterministic, that is). 
The essence of quantitative information flow, though, can be
formulated in a uniform way for all measures of entropy. 
All we have to assume is a fixed measure of entropy ${\mc H}(X)$ for a random variable $X$ and a notion of conditional entropy ${\mc H}(X|Y=y)$ indicating the entropy of $X$ under the assumption that the result of a second (possibly correlated) random experiment $Y$ is known to be $y$.

Leaks can then be measured in terms of \emph{lost entropy}, that is, the difference between the initial entropy of the secret $H$, and the remaining entropy when observing the output~$O$: ${\mc H}(H) - {\mc H}(H|O)$. 
Formally the \emph{bounding problem} of quantitative information flow is then defined as to determine whether this difference is bounded by a given constant~\cite{Smith/2009/OnTheFoundationsOfQantitativeInformationFlow,Yasuoka+Terauchi/2010/OnBoundingProblemsOfQuantitativeInformationFlow}.

\paragraph*{Min Entropy}
The min entropy measure concerns the highest probability that an attacker correctly guesses the secret in a single try and it is considered to be one of the most appropriate indications of how useful the observations are for the attacker~\cite{Smith/2009/OnTheFoundationsOfQantitativeInformationFlow}.
The min entropy of a random variable $H$ is defined as $-\log \max_h p(H=h),$
where $p$ indicates the probability measure. 
The conditional min entropy is defined as $-\log {\sum_{o}p(O=o) (\max_h p(H=h|O=o))}$, which, by Bayes' theorem, becomes $
 -\log {\sum_{o}\max_h p(O=o|H=h)p(H=h)}$. 

Under the assumption that the system under consideration is deterministic $p(O=o|H=h)$ is either 0 or 1 and further $p(H=h)$ is a constant $c_H$ if we assume the secret to be uniformly distributed. 
The bounding problem of quantitative noninterference for the min entropy measure and for a given bound $n$ hence simplifies to the following question:
\[
 $n$ ~~\geq~~ - \log c_H ~+~ \log \Big(~c_H{\sum_{o,~\exists h. p(O=o|H=h)=1} 1}~\Big),
\]
and amounts to specifying that there is no tuple of $2^n+1$ distinguishable
paths. (see~\cite{Smith/2009/OnTheFoundationsOfQantitativeInformationFlow}
for more detailed discussions.)  In a synchronous system setting
(cf.~Section~\ref{ssect:examples}), where $I\subseteq\AP$ and
$O\subseteq\AP$ refer to the atomic propositions indicating the input
and output, respectively, we can, thus formulate this condition as
follows:
\[
\varphi_{\mathit{QIF}}=\neg\exists\pi_0.\;\dots\;.~\exists\pi_{2^n}.~ \Big(\bigwedge_{i}\pi_i=_{I}\pi_0\Big) \wedge \bigwedge_{i\not=j}\pi_i\not=_{O}\pi_j
\]
Yasuoka and Terauchi proved that $2^{n+1}$ is the optimal number of
paths to formulate this property
\cite{Yasuoka+Terauchi/2010/OnBoundingProblemsOfQuantitativeInformationFlow}.

\section{Conclusions and Open Problems}
\label{sect:conclusions}

We have introduced the temporal logic HyperCTL, which extends
CTL* with path variables. HyperCTL provides a uniform logical
framework for temporal hyperproperties.  The quantification over paths
in HyperCTL subsumes related extensions of temporal logic such as the
\emph{hide} operator of SecLTL and the \emph{knowledge} operator of
temporal epistemic logic.  Beyond the properties expressible in these
temporal logics, HyperCTL expresses a rich set of properties from
security and coding theory, such as noninterference and the
preservation of entropy. The encoding of these properties in HyperCTL
is natural and efficient: for many security properties the complexity
of the HyperCTL model checking algorithm matches the lower bounds
known for the individual properties.
%
%
HyperCTL thus not only provides a common semantic foundation for
properties from a diverse class of applications but also a uniform
algorithmic solution.

An important open question concerns the identification of efficient
fragments of HyperCTL, motivated in particular by the significant gap
between the PSPACE-complete model checking problem for CTL* and the
non-elementary model checking problem of general HyperCTL
formulas. 

From a semantic perspective, it is worth noting that the
expressiveness of HyperCTL exceeds the original definition of
hyperproperties~\cite{Clarkson+Schneider/10/Hyperproperties}. While
Clarkson and Schneider's hyperproperties are sets of sets of
sequences, and thus have their semantic foundation in the linear-time
paradigm of sequences, HyperCTL integrates, like CTL*, the linear-time
view with the branching-time view. When path quantifiers occur in the
scope of temporal operators, the HyperCTL formula expresses that the
new paths branch off in the same state of the Kripke structure. A full
investigation of the linear vs branching time spectrum of hyperproperties
and the implications on the complexity, compositionality, and
usability of logics for hyperproperties are intriguing challenges for
future research.



\bibliographystyle{plain}
\bibliography{bib_short}

\newpage

\appendix

\section{Epistemic Logics Without the Synchronous Time Assumption}
\label{app:asynchEpistemic}

We now discuss the adaptations necessary to extend the encoding from the proof of Theorem~\ref{thm:synchEpistemic} to non-synchronous systems. The following proof of Theorem~\ref{thm:asynchEpistemic} is also a good example for the general treatment of asynchronous executions in HyperCTL. 

\setcounter{aux}{\value{theorem}}
\setcounter{theorem}{\value{thm-asynchEpistemic}}

\begin{theorem}
For every epistemic temporal logic formula $\psi$ and interpreted system $\I$, there  is a HyperCTL formula $\varphi$ such that 
$\I\models\psi$ iff $K'(\I)\models\varphi$. 
\end{theorem}

\setcounter{theorem}{\value{aux}}

\begin{proof}
First, we define the model transformation $K'(\I)$, which differs from $K(\I)$ in that it adds stuttering steps. In the resulting Kripke structure, we indicate for each transition whether it is due to a step of the interpreted system or if it is a newly added stuttering step.

For a given interpreted system $\I$, the stuttering Kripke structure $K'(\I)=(S,s_0,\delta,\AP,L)$ is defined via the non-stuttering Kripke structure $K(\I)=(S',s'_0,\delta',\AP',L')$, where $S=S'\times\{\text{stutter},\text{move}\}$, $s_0=(s_0,\text{stutter})$, $\delta(\langle s,x\rangle )=\delta'(s)\times\{\text{move}\}\cup \{\langle s,\text{stutter}\rangle \}$, $\AP=\AP\cup\{\text{stutter},\text{move}\}$, and $L(\langle s,x\rangle )=L'(s)\cup\{x\}$. 

In $K'(\I)$, the path quantifiers range over a larger set of paths that we have to restrict to those that do not stutter forever ($\mathsf{progress}(\pi)=\LTLsquare\LTLdiamond\text{move}_\pi$) and to those pairs of paths that synchronize correctly. 
The correct synchronization of two paths $\pi$ and $\pi'$ means that the observations regarding some set of propositions $P$ changes in one path iff it changes also in the other path:
$\mathsf{synch}(\pi,\pi')=\LTLsquare\Big( \pi[0]\not=_P\LTLcircle\pi[0] ~~\Leftrightarrow~~ \pi'[0]\not=_P\LTLcircle\pi'[0]\Big)$. 

The encoding from the proof of Theorem~\ref{thm:synchEpistemic} is now modified as follows for positively occurring knowledge operators:
\[
\begin{array}{l}
\mathbf{Q}. \exists\pi. ~\varphi'|_{{\sf K}_P\psi\to u_{\pi}}\;\wedge\;
\big(
\forall\pi'.~ (t_{\pi'}\U\LTLsquare \neg t_{\pi'}) ~\to~ \forall\pi''.~  \mathsf{progress}(\pi'') ~\wedge~ \mathsf{synch}(\pi,\pi'') ~\wedge~ \\\qquad  \LTLsquare(t_{\pi'}\to (\pi_n[0]\!=_P\!\pi''[0])) ~\to~ \LTLsquare(t_{\pi'}\wedge u_{\pi}\to [\psi]_{\pi''})\big)\;,
\end{array}
\]
and, for negatively occurring knowledge operators:
\[
\begin{array}{l}
\mathbf{Q}. \exists\pi. ~\varphi'|_{{\sf K}_P\psi\to u_{\pi}}\;\wedge\;
\big(\forall\pi'.~ (t_{\pi'}\U\LTLsquare \neg t_{\pi'}) \to \exists\pi''.~ \mathsf{progress}(\pi'') ~\wedge~ \mathsf{synch}(\pi,\pi'') ~\wedge~ \\\qquad \LTLsquare(t_{\pi'}\to (\pi_n[0]\!=_P\!\pi''[0]))  ~\wedge~ \LTLdiamond(t_{\pi'}\!\wedge u_{\pi}\wedge \neg [\psi]_{\pi''})\big)\;,
\end{array}
\]

For a positively occurring knowledge operator $\mathcal{K}_P\psi$, we prove that given two executions $\pi_1$ and $\pi_2$, if the sub-formula $\psi$ must hold on $\pi_2$ at some position $i$ (by the semantics of the knowledge operator), then the formula above requires that on a stuttered version of $\pi_1$ the sub-formula is applied at state $\pi_2[i]$. 
We consider the prefixes of the two executions, $\pi_1[0,i]$ and $\pi_2[0,j]$ and assume they have the same traces with respect to $P\subseteq\AP$. 
For these, there are two stuttered versions $\pi$ and $\pi''$ (named to match the path variables in the formula above) that synchronize the positions at which they change their observations with respect to $P$ and pad the prefixes to the same length $k$ (without changing their final states, that is $\pi_2[i]=\pi''[k]$). 
There is also a labeling with $t_{\pi'}$ that ensures that the knowledge operator is evaluated until position $k$. 
Hence, $\psi$ is applied on $\pi''$ at state $\pi_2[i]$. 

The other direction is straightforward, and the case of a negatively occurring knowledge operator follows similarly. 
\end{proof}

\section{Goguen and Meseguer's Noninterference}\label{app:noninterference}

\subsection{Noninterference}

A point of reference for most of the literature on information flow is the definition of \emph{noninterference} that was introduced by Goguen and Meseguer in 1982 \cite{Goguen+Meseguer/1982/SecurityPoliciesAndSecurityModels}. 
In this subsection, we show how to express noninterference in a simple HyperCTL formula. 

The system model used in \cite{Goguen+Meseguer/1982/SecurityPoliciesAndSecurityModels}, which we will refer to as \emph{deterministic state machines}, operates on commands $c\in C$ that are issued by different users $u\in U$. 
The evolution of a deterministic state machine is governed by the transition function $\DO:S\times U\times C\to S$ and there is a separate observation function $\out:S\times U\to \Out$ that for each user indicates what he can observe. 


We define standard notions on sequences of users and events. 
For $w\in(U\times C)^*$ and $G\subseteq U$ let $|w|_{G}$ denote the projection of $w$ to the commands issued by the users in $G$. 
Further, we extend the transition function $\DO$ to sequences, $\DO(s,(u,c).w)=\DO(\DO(s,u,c),w)$, where the dot indicates concatenation. 
Finally, we extend the observation function $\out$ to sequences $w$, indicating the observation \emph{after} $w$: $\out(w,G)=\out(\DO(s_0,w),G)$.
\emph{Noninterference} is then defined as a property on systems $M$. A set of users $G_H\subseteq U$ does not interfere with a second group of users $G_L\subseteq U$, if:
\[
\forall w\in (U\times C)^*. ~\out(w,G_L)=\out(|w|_{G_H},G_L)~.
\]

That is, we ask whether the same output would be produced by the system, if all actions issued by any user in $G_H$ were removed.

\subsection{Encoding GM's System Model}

First, we need to map their system model into Kripke structures as used for the formulation of HyperCTL (which, obviously, must not solve problem by itself). 
We choose an intuitive encoding of state machines in Kripke structures that indicates in every state (via atomic propositions) which observations can be made for the different users, and also which action was issued last, including the responsible user.

We choose a simple translation that maps a state machine  $M=(S,U,C,\Out,\out,\DO,s_0)$ to the Kripke structure $K=(S',s'_0,\delta,\AP,L)$, where $S'=\{s_0\}\cup S\times U\times C$, $s'_0=s_0$, $\AP= U \times C ~\cup~ U\times\Out$, and the labeling function is defined as $L(s_0)=\{(u,\out(s_0,u))\mid u\in U\}$ for the initial state and $L((s,u,c))=\{(u,c)\}\cup\{(u',\out(s,u'))\mid u'\in U\}$ for all other states. 

The transition function is defined as 
\[
\delta(s,u,c) = \{(s',u',c') \mid \DO(s,u',c')=s'\!,~ u'\!\in U,~ c'\!\in C \}
\]

Each state (except for the initial state) has labels indicating the command that was issued last, and the user that issued the command. 
The remaining labels denote the observations that the individual users can make in this state. 

To access these two separate pieces of information, we introduce  functions $\input:S\to U\times C$, which is not defined for $s_0$, and $\out:S\times U\to\Out$ (by abusing notation slightly), with their obvious meanings. 

In this system model let $s=_{U\setminus G_H}s'$ mean that if we entered one of the states $s$ and $s'$ with a user not in $G_H$ the users and commands must be identical. 
If both states are entered by users in $G_H$ then the commands may be different. 
The output equality on states $s=_{O,G_L}s'$, shall refer to the observations the users $G_L$ can make at this state. 
Then, noninterference can be expressed as follows:
\[
\begin{array}{l}
\varphi_{\mathit{NI}}(G_H,G_L)~=~ \forall \pi. \forall\pi'.~ \pi[0]=_{U\setminus G_H}\pi'[0] ~~\W~~
\Big( \pi[0]\not=_{U\setminus G_H}\pi'[0] ~\wedge~ \\ 
\qquad\qquad\qquad\qquad \big(~\pi[0]\in H ~\wedge~ \LTLsquare(\LTLcircle\pi[0]=_{U\setminus G_H}\pi'[0]) ~\implies~ \LTLsquare(\LTLcircle\pi[0]=_{O,G_L}\pi'[0]) ~\big)\Big) ~,
\end{array}
\]

\begin{theorem}\label{thm:noninterference_app}
There is a HyperCTL formula $\varphi_{\mathit{NI}}(G_H,G_L)$ and an encoding $K(M)$ of state machines into Kripke structures such that for every state machine $M$, and groups of users $G_H$ and $G_L$ it holds $K(M)\models\varphi_{\mathit{NI}}(G_H,G_L)$ iff $G_H$ does not interfere with $G_L$ in $M$.  
\end{theorem}

\begin{proof}
The formula pattern $\pi[0]=_{U\setminus G_H}\pi'[0] ~\W~ \big(\pi[0]\not=_{U\setminus G_H}\!\pi'[0] ~\wedge~ \varphi \big)$ implies that $\varphi$ is applied exactly at the first position at which $\pi$ and $\pi'$ differ in their input (except, possibly, on the input of users in $G_H$). 
As, for a fixed path $\pi$, we quantify over all paths $\pi'$ the subformula $\varphi$ is hence applied to every position of every path $\pi$. 

The subformula $\varphi$ then requires that, if from that position on $\pi$'s input differs from the input of $\pi'$ only in that it has an additional (secret) action by some user in $G_H$, both paths must look equivalent from the view point of the users in $G_L$. 
The interesting part here is the use of the next operator, as it enables the comparison of different positions of the traces. 
Thus, we compare path $\pi$ to all other paths that have one secret action less than itself. 
Hence, by transitivity of equivalence, we compare all paths $\pi$ to a version of itself that is stripped of all secret actions. 
\end{proof}

\section{Observational Determinism}
\label{app:OD}

\setcounter{aux}{\value{theorem}}
\setcounter{theorem}{\value{thm-obsDet}}

\begin{theorem}
  There is an encoding of programs as defined in 
  \cite{HuismanWS/06/TLCharacterisationOfOD} into Kripke
  structures, such that HyperCTL can express observational determinism. 
\end{theorem}

\setcounter{theorem}{\value{aux}}

\begin{proof}
Huisman et al.~\cite{HuismanWS/06/TLCharacterisationOfOD} define observational determinism over programs in a simple while language, which is very similar to the execution model of \cite{Zdancewic+Myers/03/ObservationalDeterminism}. 
First, they define a special version of self-composition on their programs that allows both sides to move independently by introduces stuttering steps. 
Then, they encode programs into Kripke structures that do not only maintain the current state, but also remember the last state. 
The CTL* formula is then proved to precisely express observational determinism. 
Since HyperCTL subsumes CTL*, we can make use of the same encoding of programs into Kripke structures, but we leave out the self-composition operation on programs. 
In order to re-enable the correct synchronization of paths, however, we have to introduce stutter steps for single, not self-composed that is, programs. 
%

We prepend the CTL* formula with two path quantifiers over paths $\pi_1$ and $\pi_2$ and we let those propositions of the formula that referred to the two copies of the program now refer to the two paths, respectively. 
\end{proof}

\section{An Example of a Declassification Property}
\label{app:declassification}

Although we have shown that declassification properties can be encoded in principle, the two-step encoding via epistemic logics, will not necessarily lead to a practical verification approach. 
Further, the encoding via epistemic logics~\cite{Balliu+Dam+Guernic/2011/EpistemicTemporalLogicForIFSecurity} assumes a non-reactive system model (it has no inputs) and thus its flavor of noninterference on which they build the other encodings does not subsume Goguen and Meseguer's definition. 
In this appendix we show with the help of an example that declassification properties can be encoded in an elegant way leading to efficient algorithmic approaches. 

Let us continue to work in the setting of deterministic state machines as we used already for noninterference~\cite{Goguen+Meseguer/1982/SecurityPoliciesAndSecurityModels}. 
We consider the case of a system that shall allow a user $A$ to set a secret password for her account names ``Alice'', which a user $B$ wants to access. 
In order to get access to the account $B$ would have to guess the password correctly, which we \emph{assume} impossible without additional knowledge. 
We would like to specify the \emph{external behavior} of the system in terms of the commands that allow to change passwords and access accounts. 
We assume that passwords are stored as hash values of length 256. 

Let $C \supseteq \{\mathit{changePwd}\}\times\mathit{Accounts}\times\mathbb{B}^{256} ~\cup~ \{\mathit{login}\}\times\mathit{Accounts}\times\mathbb{B}^{256}$ and $U\supseteq\{A,B\}$, where $\mathit{Accounts}$ refers to the domain of account names.
We assume that passwords can only be changed while logged in and that user $A$ is signed in from the beginning. 
Further, let $i\in\mathbb{B}^{256}$ be the initial password. 
Then the formulation of noninterference could be refined to compare only those paths on which the current password was not guessed correctly. 
First, we define the LTL property on individual paths that the initial password is not guessed until it is replaced, and that the same holds for every password that is set during the run of the system:
\[
\begin{array}{ll}
\psi=\!\!\!\!& \bigwedge_{x\in\mathbb{B}^{256}}(\neg\langle B,\langle\mathit{login},\mathit{Alice},i\rangle\rangle \U \langle A,\langle\mathit{changePwd},\mathit{Alice},x\rangle\rangle ) ~\wedge~\\
&\qquad\qquad\qquad\qquad \LTLsquare\big(\bigwedge_{x,y\in\mathbb{B}^{256}} \langle A,\langle\mathit{changePwd},\mathit{Alice},x\rangle\rangle ~\implies~ \\
&\qquad\qquad\qquad\qquad\qquad\neg\langle B,\langle\mathit{login},\mathit{Alice},x\rangle\rangle\U \langle A,\langle\mathit{changePwd},\mathit{Alice},y\rangle\rangle\big) ~,	
\end{array}
\]
and then embed $\psi$ into $\varphi_{\mathit{NI}}$:
\[
\begin{array}{l}
\forall \pi. \forall\pi'.~[\psi]_{\pi}\wedge[\psi]_{\pi'} \implies ~\pi[0]=_{U\setminus G_H}\pi'[0] ~~\W~~
\Big( \pi[0]\not=_{U\setminus G_H}\pi'[0] ~\wedge~ \\ 
\qquad\qquad\qquad\qquad \big(~\pi'[0]\in H ~\wedge~ \LTLsquare(\pi[0]=_{U\setminus G_H}\LTLcircle\pi'[0]) ~\implies~ \LTLsquare(\pi[0]=_{O,G_L}\LTLcircle\pi'[0]) ~\big)\Big) ~.
\end{array}
\]


Adding the assumption that the password is not simply guessed correctly essentially did not change the structure of the property. 
The advantage of the logical approach is that such assumptions could also be added to other definitions of secrecy, like that of observational determinism, without the need to reengineer the algorithm for their verification completely. 

\end{document}